\documentclass[16pt]{amsart}
\usepackage{graphicx,amssymb,amsmath,amsthm}
    \usepackage[titletoc]{appendix}
\usepackage{latexsym}
\usepackage{comment,cite,color}
\usepackage{cite,color}
\usepackage{algorithmic}
\usepackage{enumerate}
\usepackage{mathrsfs}
\usepackage[ruled]{algorithm}
\usepackage{epsfig}

\newtheorem{thm}{Theorem}

\newtheorem{cor}{Corollary}
\newtheorem{lem}{Lemma}

\theoremstyle{definition}
\newtheorem{defn}{Definition}

\newtheorem{remark}{Remark}

\newcommand{\x}{\mathbf x}
\newcommand{\y}{\mathbf y}
\newcommand{\bu}{\mathbf u}
\newcommand{\bv}{\mathbf v}
\newcommand{\V}{\mathbf V}

\newcommand{\C}{\mathbb{C}}
\newcommand{\Z}{\mathbb{Z}}

\newcommand{\OMMP}{{\rm OMMP}}

\newcommand{\argmin}[1]{\mathop{\rm argmin}\limits_{#1}}
\newcommand{\abs}[1]{\lvert#1\rvert}
\newcommand{\aarg}[1]{\underset{#1}{\rm argmin}}
\newcommand{\innerp}[1]{\langle {#1} \rangle}
\newcommand{\myceil}[1]{\left\lceil {#1}\right\rceil}
\newcommand{\myfloor}[1]{\left\lfloor {#1}\right\rfloor}
\date{}

\begin{document}
\bibliographystyle{plain}
\title{the performance of orthogonal multi-matching
pursuit under RIP }
\author{  Zhiqiang Xu}
\thanks{Supported by the National Natural Science Foundation of China (11171336).}
 \maketitle

\begin{abstract}
The orthogonal multi-matching pursuit (OMMP) is a natural extension of the orthogonal matching pursuit (OMP).
We denote the OMMP with the parameter $M$ as $\OMMP(M)$ where $M\geq 1$ is an integer.
The main difference between OMP and $\OMMP(M)$ is that $\OMMP(M)$ selects $M$ atoms per iteration,
while OMP only adds one atom to the optimal atom set. In this paper, we study the performance of orthogonal
multi-matching pursuit under RIP. In particular, we show that, when the measurement matrix $A$ satisfies $(9s, 1/10)$-RIP, there exists an absolute  constant $M_0\leq 8$ so that
$\OMMP(M_0)$ can recover $s$-sparse signal within $s$ iterations.  We furthermore prove that
  $\OMMP(M)$ can recover $s$-sparse signal within $O({s}/{M})$ iterations for a large class of $M$ provided the signal is slowly-decaying. In particular,  for $M=s^a$ with $a\in [0,1/2]$, $\OMMP(M)$ can recover slowly-decaying $s$-sparse signals within $O(s^{1-a})$ iterations. The result implies that $\OMMP$ can reduce the computational
complexity heavily.

\end{abstract}

\section{Introduction}

\subsection{Orthogonal Matching Pursuit }

Orthogonal matching pursuit (OMP) is a popular algorithm for the recovery of sparse signals and it is also commonly used in compressed sensing. Let $A$ be a matrix of size $m\times N$ and $\y$
 be a vector of size $m$.
 The aim of OMP is to find the approximate solution to the following $\ell_0$-minimization problem:
 $$
 \min_{\x\in \C^N}\|\x\|_0 \qquad {\rm s.t.}\qquad A\x=\y,
 $$
 where $\|\x\|_0$ denotes the number of non-zero entries in $\x$.
In compressed sensing and the sparse representation of signals, we often have $m\ll N$. Throughout this paper, we  suppose that the sampling
matrix $A\in \C^{m\times N}$ whose columns $a_1,\ldots,a_N$ are $\ell_2$-normalized.

To introduce the performance of OMP, we first recall the definition of the restricted isometry property (RIP)  \cite{ctrip} which is frequently used in the analysis of the recovering algorithm in compressed sensing.  Following Cand\`{e}s and Tao, for $1\leq s\leq N$ and $\delta\in [0,1)$,
we say that the matrix $A$ satisfies $(s,\delta)$-RIP  if
\begin{equation}\label{eq:con}
(1-\delta) \|\x\|_2^2 \leq \|A \x\|_2^2 \leq (1+\delta) \|\x\|_2^2
\end{equation}
holds for all $s$-sparse signals $\x$.
We say that the signal $\x$ is {\em $s$-sparse} if $\|\x\|_0\leq s$
and use $\Sigma_s$ to denote the set of $s$-sparse signals, i.e.,
$$
\Sigma_s\,\,=\,\,\{\x\in \C^N:\|\x\|_0\leq s\}.
$$
We next state the definition of the spark (see also \cite{spark}).
\begin{defn}
The spark of a matrix $A$ is the size of the smallest linearly dependent subset of columns, i.e.,
$$
{\rm Spark}(A):=\min\{\|\x\|_0:A\x=0, \x\neq 0\}.
$$
\end{defn}

Theoretical analysis of OMP has concentrated primarily on two directions. The first one is to study the condition for the matrix $A$ under which OMP can recover $s$-sparse signals in  exactly $s$ iterations. In this direction, one uses the coherence and RIP to analyze the performance of OMP.  In particular, Davenport and Wakin showed that, when the matrix $A$ satisfies $(s+1, \frac{1}{3\sqrt{s}})$-RIP,  OMP can recover $s$-sparse signal in exactly $s$ iterations \cite{ripomp}.
The sufficient condition is improved to $({s+1}, \frac{1}{\sqrt{s}+1})$-RIP in \cite{rayomp, moomp} (see also \cite{liuomp, huangomp} ).
However, it was observed in \cite{rauhutomp}, when the matrix $A$ satisfies $(c_0s, \delta_0)$-RIP for some fixed constants $c_0>1$ and $0<\delta_0<1$, that $s$ iterations of OMP is not enough to uniformly  recover   $s$-sparse signals, which implies that OMP has to run for more than  $s$ iterations to uniformly recover the $s$-sparse signals. Hence, one investigates the performance of OMP along the second line with  allowing to  OMP run more than $s$ iterations. For this case, it is possible that OMP  add wrong atoms to the optimal atom set, but one can identify the correct atoms by  the least square. A main result in this direction is presented by Zhang \cite{zhangomp} with proving that when $A$ satisfies $(31s, {1}/{3})$-RIP OMP can recover the $s$-sparse signal in at most $30s$ iterations.

 The other type of greedy algorithms, which are based on OMP, have been proposed including the regularized orthogonal matching pursuit (ROMP) \cite{romp}, subspace pursuit (SP) \cite{subspace}, CoSaMP \cite{cosamp}, and many other variants. For each of these algorithms, it has been shown that, under a natural RIP setting, they can recover the $s$-sparse signals in $s$ iterations.

\subsection{Orthogonal Multi-matching Pursuit and Main Results }
A more natural extension of OMP is the orthogonal multi-matching pursuit (OMMP) \cite{liuomp}. We denote the OMMP with the parameter $M$ as $\OMMP(M)$ where $M\geq 1$ is an integer. The main difference between OMP and $\OMMP(M)$ is that $\OMMP(M)$ selects $M$ atoms per iteration, while OMP only adds one atom to the optimal atom set.  The Algorithm 1 outlines the procedure of $\OMMP(M)$ with initial feature set  $\Lambda^0$. In comparision with OMP, OMMP has fewer iterations and computational complexity \cite{huangomp}. We note that, when $M=1$, $\OMMP(M)$ is identical to OMP. OMMP is also studied in \cite{rayomp, huangomp, gomp} under the names of KOMP, MOMP and gOMP, respectively. These results show that, when RIP constant $\delta=O(\sqrt{{{M}}/{{s}}})$, $\OMMP(M)$ can recover the $s$-sparse signal in $s$ iterations.

\begin{algorithm}
\begin{algorithmic}
 \STATE {\bf Input:} sampling matrix $A$, samples $\y=A\x$, candidate number $M$ for each step, stopping iteration index $H$, initial feature set  $\Lambda^0\subset \{1,\ldots,N\}$
 \STATE {\bf
Output:} the $\x^*$.
 \STATE {\bf Initialize:} $ \ell=0$.
 \STATE ${\mathbf x}^0=\aarg{{\mathbf z}: {\rm supp}({\mathbf z})\subset \Lambda^{0}}\|\y-A {\mathbf z}\|_2, {\mathbf r}^0={\mathbf y}-A \x^0$
  \WHILE {  $\ell< H$}
  \STATE {\bf match:} ${\mathbf h}^\ell=A^T{\mathbf r}^\ell$
  \STATE{\bf calculate:} $T^\ell=$
  \text{
  $M$ indices } \text{corresponding  to the largest  magnitude } \text{ entries in the vector  $h^\ell$
     }
  \STATE{\bf identity:} $\Lambda^{\ell+1}=\Lambda^\ell\cup T^\ell$
  \STATE{\bf update:}  ${\mathbf x}^{\ell+1}=\aarg{{\mathbf z}: {\rm supp}({\mathbf z})\subset \Lambda^{\ell+1}}\|\y-A {\mathbf z}\|_2$
  \STATE \,\,\qquad\qquad ${\mathbf r}^{\ell+1}=\y-A {\mathbf x}^{\ell+1}$
  \STATE \,\,\qquad\qquad $\ell=\ell+1$
 \ENDWHILE
 \STATE \,\, $\x^*={\mathbf x}^{H}$
\end{algorithmic}
\caption{\small{$\OMMP(M)$}}
\end{algorithm}

 The aim of this paper is to study the performance of $\OMMP(M)$ under a more natural setting of RIP (the RIP constant is an absolute constant). Particularly, we also would like to understand the relation between the number of iterations and the parameter $M$.
So, we are interested in the following questions:
\begin{enumerate}[{\bf Question} 1]
\item  {\em Does there exist an absolute constant $M_0$ so that $\OMMP(M_0)$ can recover all the $s$-sparse signals within $s$ iterations?}\newline
\item  {\em For $1\leq M\leq s$, can $\OMMP(M)$ recover the $s$-sparse signals within $O({s}/{M})$ iterations?}
\end{enumerate}

We next state one of our main results which gives an affirmative answer to Question 1.

\begin{thm}\label{th:ommpt0}
Let $\x\in \Sigma_s$ and $S={\rm supp}(\x)$.  Suppose that the sampling
matrix $A\in \C^{m\times N}$  satisfies $(9s,{1}/{10})$-RIP and  ${\rm Spark}(A)>\max\{Ms', 8s'\}+\#\Lambda^0 $ where
 $\Lambda^0$ is the initial feature set in $\OMMP$ algorithm.  Then
$\OMMP(M)$ can recover the signal $\x$ within, at most,
 $\max\{s',\frac{8}{M}s'\}$ iterations, where $s':=\# (S\setminus \Lambda^0)$.
\end{thm}

The above theorem shows that, when $M\geq 8$, $\OMMP(M)$  with the initial feature set $\Lambda^0=\emptyset$  can recover all the $s$-sparse signal within, at most, $s$ iterations. It implies that there exists an absolute constant $M_0\leq 8$ so that $\OMMP(M_0)$ can recover all the $s$-sparse signals within $s$ iterations. We believe that the constant $M_0=8$ is not optimal. The numerical experiments make us conjecture that the optimal number is $2$, i.e., under RIP, $\OMMP(2)$ can recover the $s$-sparse signal within $s$ iterations.

We next turn to Question 2. The following theorem shows that, when $1\leq M\leq \sqrt{s}$, $\OMMP(M)$ can recover slowly-decaying signal within $O({s}/{M})$ iterations.

\begin{thm}\label{th:ommpt1}
Let $\x\in \Sigma_s$, $S={\rm supp}(\x)$ and $s'=\#(S\setminus \Lambda^0)$.
Consider the $\OMMP(M)$ algorithm with  $1\leq M\leq \sqrt{s'}$ and the initial feature set $\Lambda^0$.
If the sampling matrix $A\in \C^{m\times N}$  satisfies $(9s, \frac{1}{10})$-RIP and
 $$
 {\rm Spark}(A)> 8(C_0^2+2){s'}+\#\Lambda^0,
 $$
    then $\OMMP(M)$ recovers the  $\x$ within  $\myfloor{8(C_0^2+2){s'}/{M}}$ iterations where
$
C_0={\max\limits_{j\in S }\abs{\x_j}}/{\min\limits_{j\in S}\abs{\x_j}}.
$

\end{thm}

The theorem above shows that, for $1\leq M\leq \sqrt{s}$,  $\OMMP(M)$ can recover $s$-sparse signals within  $C_1 {s}/{M}$  iterations. Here, the constant $C_1$ depends on the signal $\x$. In particular, if we take  $M=\lfloor{s}^{a}\rfloor$  in Theorem {\rm\ref{th:ommpt1},} we have

\begin{cor}  Under the condition of Theorem {\rm\ref{th:ommpt1},}
if  $M=\lfloor{s}^{a}\rfloor$ with $a\in [0,1/2]$, then $\OMMP(M)$ with the initial feature set $\Lambda^0=\emptyset$ recovers the $s$-sparse signal within $\myfloor{8(C_0^2+2)s^{1-a}}$ iterations.
\end{cor}

We next consider the case with $M=\alpha \cdot s$. In particular, for `small' $\alpha$,
we give an affirmative answer to Question 2 up to a log factor.

\begin{thm}\label{th:ommpt2}
Let $\x\in \Sigma_s$ and  $S={\rm supp}(\x)$.
Suppose that the sampling
matrix $A\in \C^{m\times N}$ satisfies $(14s,{1}/{10})$-RIP and
 $$
 {\rm Spark}(A)> 8s\log_2(2(s+1)).
 $$
Consider the $\OMMP(M)$ algorithm with  the initial feature set $\Lambda^0=\emptyset$. If $M=\alpha \cdot s$, then $\OMMP(M)$ recover the $s$-sparse signal $\x$  from $\y=A\x$ within
 $\myceil{\frac{8}{\alpha}\log_2(2(s+1))}$ iterations, where   $0< \alpha\leq {2}/{(C_0^2+2)}$ and
 $
C_0={\max\limits_{j\in S }\abs{\x_j}}/{\min\limits_{j\in S}\abs{\x_j}}.
$
\end{thm}

\begin{remark}
We prove the main results
using some of the techniques developed by Zhang in his study of
OMP \cite{zhangomp} (see also \cite{foucartomp}). To make the paper more readable,  we state our results for the strictly sparse signal.  In fact, using a similar method, one also can extend the results in this paper to the case where the measurement vector $\y$ is subjected to an additive noise and $\x$ is not strictly sparse.
\end{remark}

\begin{remark}
In \cite{liuomp}, Liu and Tymlyakov proved that, when $A$ satisfies $(M_0, \delta)$-RIP with $\delta={\sqrt{M_0}}/{((2+\sqrt{2})\sqrt{s})}$, $\OMMP(M_0)$ can recover $s$-sparse signal within, at most, $s$ iterations. The result requires the RIP constant $\delta$ depends on $s=\|\x\|_0$. In Theorem \ref{th:ommpt0}, we require that the measurement matrix $A$ satisfies $(9s,\delta)$-RIP with $\delta$ being an absolute constant $1/10$. Hence, Theorem \ref{th:ommpt0} gives an affirmative answer to Question 1 under the more natural setting for the measurement matrix $A$.
\end{remark}

\begin{remark}
It is of interest to know which matrices $A$ obey the $(s,\delta)$-RIP and the ${\rm Spark}(A)>K$  where $K$ is a fixed constant. Much is known about finding matrices that satisfy the $(s,\delta)$-RIP (see \cite{bound,  cai1,caixu, vershynin, xu}). If we draw a random $m\times N$ matrix $A$ whose entries are i.i.d. Gaussian random variables, then ${\rm Spark}(A)=m$ with probability $1$ (see \cite{spark,spark1}). Moreover, the random matrix  $A$ also satisfies $(s,\delta)$-RIP with high probability provided
$$
m=O\left(\frac{s\log(N/s)}{\delta^2}\right).
$$
So, to make  the random matrices $A$ obey the $(s,\delta)$-RIP and the ${\rm Spark}(A)>K$, one can take
$$m=\max\left\{O\left(\frac{s\log(N/s)}{\delta^2}\right),K+1 \right\}.$$
\end{remark}
\vspace{0.5cm}

\section{Numerical experiments}

The purpose of the experiment is the comparison for the reconstruction performances of and the
 iteration number of $\OMMP(M)$ with different parameter $M$. Given the parameters $m=300$ and $N=1,500$, we randomly generate a $m\times N$ sampling matrix $A$ from the standard i.i.d Gaussian ensemble.
The support set $S$ of the sparse signal $\x$ is drawn from the uniform
distribution over the set of all subsets of $[1,N]\cap \Z$ of size $s$.
We then generate the sparse signal  $\x$  according to the probability model:
 the entries $\x_j, j\in S,$ are independent random
variable having the Gaussian distribution with mean $5$ and standard deviation $1$.

We apply the $\OMMP(M)$ to recover the sparse signal $\x$ from $\y=A\x$ for different parameters $M\in \{1, \lfloor \sqrt{s}\rfloor, \lfloor \frac{s}{2}\rfloor\}$. Note that when $M=1$, $ \OMMP(M)$ is identical with OMP.  We repeat the experiment 200 times for each number $s\in \{1,2,\ldots,80\}$ and calculate the success rate. When $\OMMP$ succeeds, we record the number of the iteration steps.
 The left graph in Fig. 1   depicts the success rate  of the reconstructing algorithm $\OMMP(M)$ with  $ M \in \{1, \lfloor \sqrt{s}\rfloor, \lfloor \frac{s}{2}\rfloor\}$. The number of the average iteration steps of $\OMMP(M)$ with  $ M \in \{1, \lfloor \sqrt{s}\rfloor, \lfloor \frac{s}{2}\rfloor\}$ are illustrated in the right graph in  Fig. 1. The numerical results show that the performance of $\OMMP(M), M\in \{\lfloor \sqrt{s}\rfloor, \lfloor \frac{s}{2}\rfloor\}$, is similar with that of OMP, while the number of iteration steps of $\OMMP(M),  M \in \{\lfloor \sqrt{s}\rfloor, \lfloor \frac{s}{2}\rfloor\}$, is far less than that of OMP, which agrees with the theoretical results presented in this paper.

  \begin{figure}[!ht]
\begin{center}
\epsfxsize=5cm\epsfbox{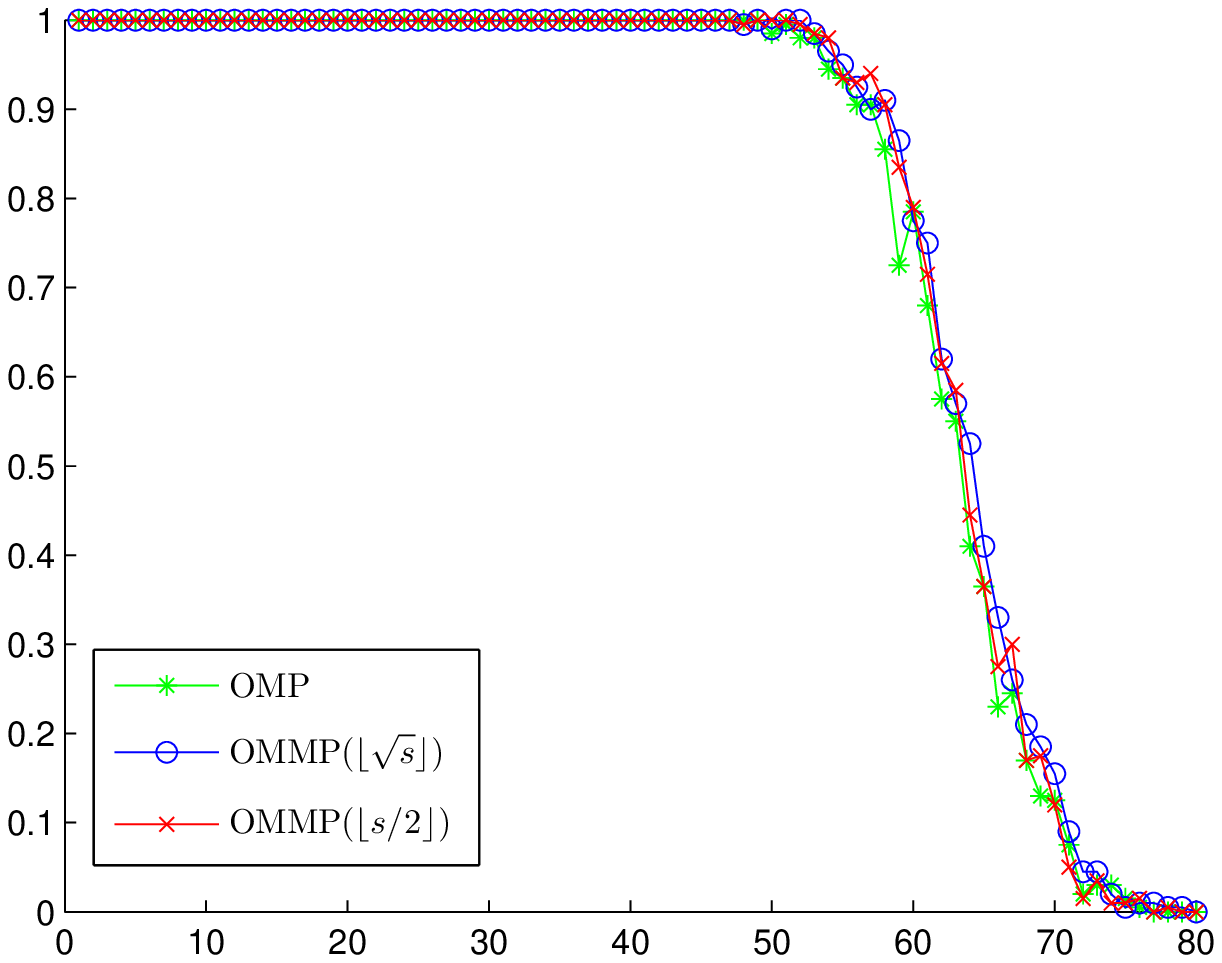}\,\,\,\,\,\,\,\,\, \epsfxsize=5cm\epsfbox{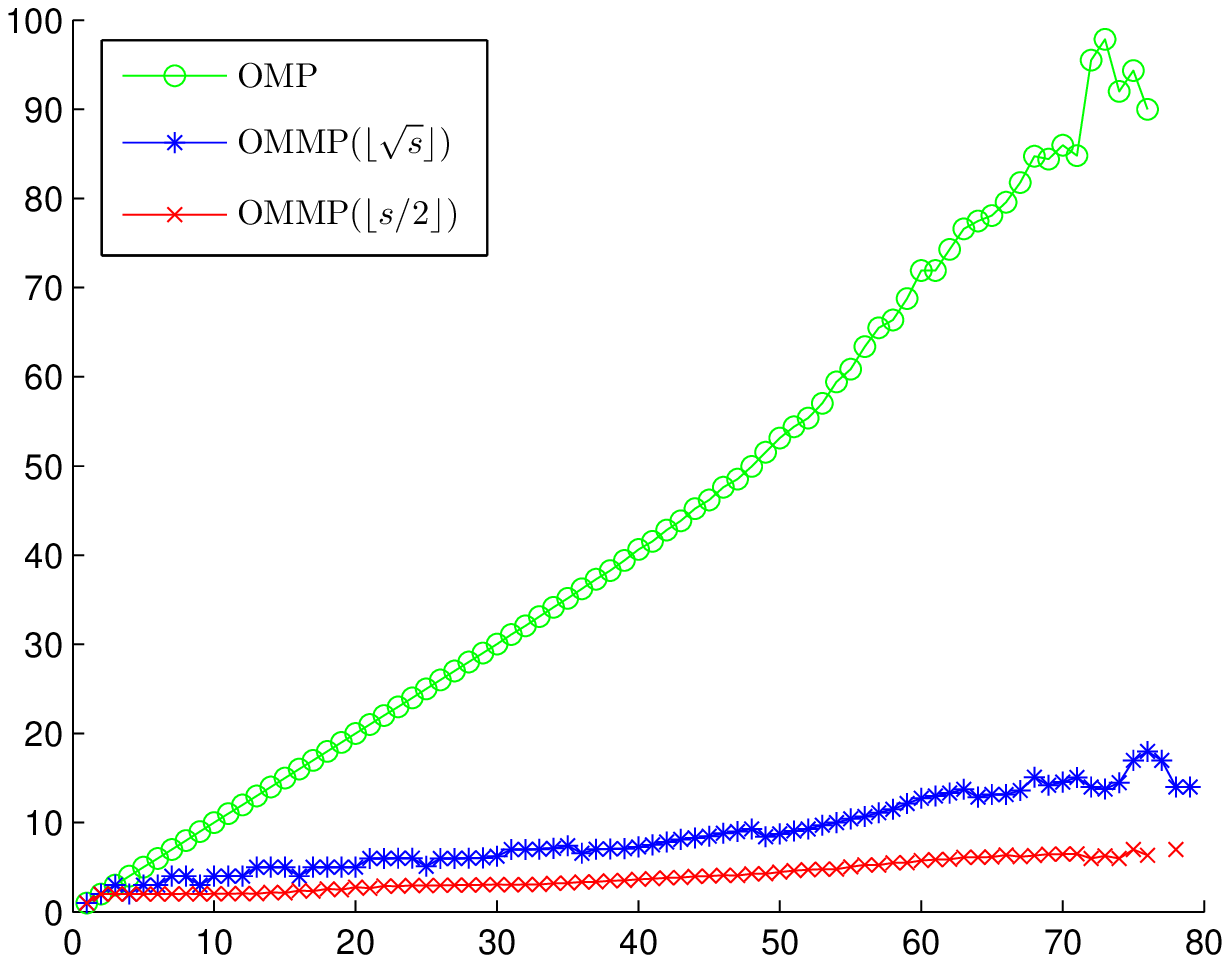}
\end{center}
\caption{\small{Numerical experiments for the  sparse signals.  The left graph corresponds to the success rates of $\OMMP(M), M\in \{1, \lfloor \sqrt{s}\rfloor, \lfloor \frac{s}{2}\rfloor\} $, whereas the right one depicts the number of the average iteration steps of   $\OMMP(M), M\in \{1, \lfloor \sqrt{s}\rfloor, \lfloor \frac{s}{2}\rfloor\} $. }}

\end{figure}

\section{Extension}

According to Theorem \ref{th:ommpt1} and Theorem \ref{th:ommpt2}, $\OMMP$ has a good performance for the slowly-decaying sparse signal $\x$. Naturally, one  may want to know whether $\OMMP(M)$ can recover all the $s$-sparse signal within less than $s$ iterations for some $M\in [1, s]\cap \Z$.  Numerical experiments show that, for some fast-decaying $s$-sparse signal $\x$, $\OMMP(M)$ has to run at least $s$ steps to recover $\x$ for any  $M\in [1, s]\cap \Z$.  However, as shown in \cite{ripomp}, when the $s$-sparse signal $\x$ is fast-decaying, OMP has a good performance. To state the result in \cite{ripomp}, we firstly introduce the definition of {\em $\alpha$-decaying} signals.
For any $s$-sparse signal $\x\in \C^N$, we denote by $S$ the support of $\x$. Without loss of generality, we suppose that $S=\{j_1,\ldots,j_s\}$ and
$$
\abs{\x_{j_1}}\geq \abs{\x_{j_2}}\geq \cdots \geq \abs{\x_{j_s}} > 0.
$$
For $\alpha>1$, we call the $\x$ {\em $\alpha$-decaying} if $\abs{\x_{j_t}}/\abs{\x_{j_{t+1}}}\geq \alpha$ for all $t\in \{1,2,\ldots,s-1\}$.

\begin{thm}{\rm (\cite{ripomp})}\label{th:omps}
Suppose that $A$ satisfies $(s+1,\delta_{s+1})$-RIP with $\delta_{s+1}<\frac{1}{3}$. Suppose that $\x$ with $\|\x\|_0\leq s$ is $\alpha$-decaying signal.
If
\begin{equation}\label{eq:alphalow}
\alpha\,\,>\,\,\frac{1+2\frac{\delta_{s+1}}{1-\delta_{s+1}}\sqrt{s-1}}{1-2\frac{\delta_{s+1}}{1-\delta_{s+1}}},
\end{equation}
then OMP will recover $\x$ exactly from $\y=A\x$ in $s$ iterations.
\end{thm}
In this paper, motivated by the proof of Theorem \ref{th:ommpt0}, we  can improve Theorem \ref{th:omps} as follows:
\begin{thm}\label{th:ompss}
Suppose that $A$ satisfies $(s,\delta_s)$-RIP with $\delta_s <\sqrt{2}-1$. Suppose that $\x\in \C^N$ with $\|\x\|_0\leq s$ is $\alpha$-decaying. If
\begin{equation}\label{eq:alphaup}
 \alpha\,\, >\,\, \sqrt{\frac{1+\delta_s}{2-(1+\delta_s)^2}},
\end{equation}
then OMP can recover $\x$ exactly from $\y=A\x$ in $s$ iterations.
\end{thm}

\begin{remark}
In Theorem \ref{th:omps}, the right side of (\ref{eq:alphalow}) depends on RIP constant and  $s=\|\x\|_0$, while in Theorem \ref{th:ompss}, the right side of (\ref{eq:alphaup}) only depends on the RIP constant. So, Theorem \ref{th:ompss} is an improvement over Theorem \ref{th:omps}.
\end{remark}

 \begin{appendices}

    \renewcommand{\thesection}{{\bf Appendix \Alph{section}}}

\section{Lemmas}
In this section, we introduce many lemmas, which extend some results in \cite{foucartomp}.
To state conveniently, for any set $T\subset\{1,\ldots,N\}$ of column indices,
we denote by $A_T$ the $m\times \# T$ matrix composed of these
columns. Similarly, for a vector $\x\in \C^N$, we use $\x_T$ to denote
the vector formed by the entries of $\x$ with indices from $T$.  For $\bu\in \C^N$ and $t\in \Z_+$, we extend the $\ell_1$-norm to a generalized $\ell_1$-norm defined as

\begin{eqnarray*}
\|\bu\|_{t,1}:=\sum_{j=0}^{\myfloor{N/t}-1} \sqrt{\bu_{jt+1}^2+\cdots
+\bu_{(j+1)t}^2}
+\sqrt{\bu_{n_0t+1}^2+\cdots+\bu_N^2}.
\end{eqnarray*}
Similarly, we also can extend the $\ell_\infty$-norm as follows
\begin{eqnarray*}
\|\bu\|_{t,\infty}&:=&\max\left\{\max_{0\leq j\leq \myfloor{N/t}-1} \sqrt{\bu_{jt+1}^2+\cdots+\bu_{(j+1)t}^2}\right.,
 \left. \sqrt{\bu_{n_0t+1}^2+\cdots+\bu_N^2}\right\}.
\end{eqnarray*}
Then the following lemma presents some inequalities for the extension norm:
\begin{lem}\label{le:cauch}
Suppose that $\bu\in \C^N$, $\bv\in \C^N$ and $t\in \Z_+$. Then
\begin{enumerate}[{\rm (i)}]
\item $${\mathcal R}(\innerp{\bu,\bv})\leq \|\bu\|_{t,\infty}\cdot  \|\bv\|_{t,1},$$
where ${\mathcal R}(\cdot)$ denotes the real part;
\item  $$\|\bu\|_{t,1}^2\leq \left\lceil\frac{N}{t}\right\rceil\cdot \|\bu\|_2^2.$$
\end{enumerate}
\end{lem}
\begin{proof}
To state conveniently, we set
$T_j:=\{j\cdot t,\ldots,j\cdot t+t\},\,  j=0,\ldots,n_0-1$ and $T_{n_0}:=\{n_0t+1,\ldots,N\}$, where
$n_0=\lfloor\frac{N}{t}\rfloor$.
 Then
\begin{eqnarray*}
{\mathcal R}(\innerp{\bu,\bv})&=&\sum_{j=0}^{n_0}{\mathcal R}(\innerp{\bu_{T_j},\bv_{T_j}})\\
&\leq& \sum_{j=0}^{n_0} \|\bu_{T_j}\|_2\cdot\|\bv_{T_j}\|_2\leq \|\bu\|_{t,\infty}\sum_{j=0}^{n_0}\|\bv_{T_j}\|_2\\
&\leq &\|\bu\|_{t,\infty}\cdot \|\bv\|_{t,1}.
\end{eqnarray*}
We now consider (ii). Note that
$$
\|\bu\|_2^2\,\,=\,\, \sum_{j=0}^{n_0}\|\bu_{T_j}\|_2^2.
$$
Then Cauchy-Schwarz inequality implies that
$$
\left\lceil\frac{N}{t}\right\rceil\cdot \|\bu\|_2^2 =\left\lceil\frac{N}{t}\right\rceil\cdot \sum_{j=0}^{n_0}\|\bu_{T_j}\|_2^2\,\,\leq \,\, \left(\sum_{j=0}^{n_0}\|\bu_{T_j}\|_2\right)^2\,\,=\,\,\|\bu\|_{t,1}^2.
$$
\end{proof}

\begin{lem}\label{th:1}
Suppose that $\Lambda^n\subset \Lambda^{n+1}\subseteq \{1,\ldots,N\}$
and set $T^n:=\Lambda^{n+1}\setminus \Lambda^n$ with $t:=\# T^n$.
 Suppose that the sampling matrix $A\in \C^{m\times N}$ satisfies  ${\rm Spark}(A)>\# \Lambda^{n+1}$.
Let
\begin{eqnarray}
\x^{n}&:=& \argmin{{\rm supp}(\mathbf z)\subset\Lambda^{n}}\|\y-A{\mathbf z}\|_2,\nonumber\\
 \x^{n+1}&:=& \argmin{{\rm supp}(\mathbf z)\subset \Lambda^{n+1}}\|\y-A{\mathbf z}\|_2,\label{eq:xndef}
\end{eqnarray}
and
$$
\V^n\,\,:=\,\,A_{T^n}^H(\y-A\x^n),
$$
where $A_{T^n}^H:=(A_{T^n})^H$.
Then
$$
\|\y-A\x^{n+1}\|_2^2\,\, \leq \,\, \|\y-A\x^n\|_2^2-\frac{1}{1+\delta_t}\|\V^n\|_2^2.
$$
\end{lem}
\begin{proof}
The definition of $\x^{n+1}$ implies that the residuality $\y-A\x^{n+1}$ is orthogonal to the space ${\rm span}(A_{\Lambda^{n+1}})$. Noting
$A(\x^{n+1}-\x^n)\in {\rm span}(A_{\Lambda^{n+1}})$, we obtain that
$$
\left< \y-A\x^{n+1}, A(\x^{n+1}-\x^n)\right>=0,
$$
which implies that
\begin{eqnarray}
\|\y-A\x^n\|_2^2 &=&\|\y-A\x^{n+1}+A(\x^{n+1}-\x^n)\|_2^2\nonumber\\
&=&\|\y-A\x^{n+1}\|_2^2+\|A(\x^{n+1}-\x^n)\|_2^2.\label{eq:nn}
\end{eqnarray}
Furthermore,
$A^H_{\Lambda^{n+1}}(\y-A\x^n)=0$ implies that
\begin{equation}\label{eq:cc1}
(A^H\y)_{\Lambda^{n+1}}=(A^HA\x^{n+1})_{\Lambda^{n+1}}.
\end{equation}
Similarly, we have
\begin{equation}\label{eq:cc2}
(A^H\y)_{\Lambda^n}=(A^HA\x^n)_{\Lambda^n}.
\end{equation}
According to  (\ref{eq:cc1}),  we obtain that
\begin{equation}\label{eq:lamn1}
(A^HA(\x^{n+1}-\x^n))_{\Lambda^{n+1}}=(A^H(\y-A\x^n))_{\Lambda^{n+1}},
\end{equation}
since
\begin{eqnarray*}
(A^HA(\x^{n+1}-\x^n))_{\Lambda^{n+1}}
=(A^H\y)_{\Lambda^{n+1}}-(A^H A\x^n)_{\Lambda^{n+1}}
 =(A^H(\y-A\x^{n}))_{\Lambda^{n+1}}.
\end{eqnarray*}
To this end, we consider
\begin{eqnarray}
\|A(\x^{n+1}-\x^n)\|_2^2 &=&\left<\x^{n+1}-\x^n,A^HA(\x^{n+1}-\x^n)\right>\label{eq:Axn2}\\
&=&\left<(\x^{n+1}-\x^n)_{\Lambda^{n+1}}, (A^HA(\x^{n+1}-\x^n))_{\Lambda^{n+1}}\right>\nonumber\\
&=&\left<(\x^{n+1}-\x^n)_{\Lambda^{n+1}}, (A^H(\y-A\x^n))_{\Lambda^{n+1}}\right>\nonumber\\
&=&\left<(\x^{n+1}-\x^n)_{T^n}, (A^H(\y-A\x^n))_{T^n}\right>\nonumber\\
&=&\left<(\x^{n+1})_{T^n}, (A^H(\y-A\x^n))_{T^n}\right>\nonumber\\
&=&\left<(A^H(\y-A\x^n))_{T^n},(\x^{n+1})_{T^n}\right>\nonumber\\
&=&\left<\V^n,(\x^{n+1})_{T^n}\right>,\nonumber
\end{eqnarray}
where the third and the fourth equality follow from (\ref{eq:lamn1}) and (\ref{eq:cc2}), respectively.
According to \eqref{eq:xndef},
$$
\x^{n+1}=A_{\Lambda^{n+1}}^+\y,
$$
where $A_{\Lambda^{n+1}}^+=(A_{\Lambda^{n+1}}^HA_{\Lambda^{n+1}})^{-1}A_{\Lambda^{n+1}}^H$ is the Moore-Penrose pseudoinverse of $A_{\Lambda^{n+1}}$. And hence
$$
\x^{n+1}=(A_{\Lambda^{n+1}}^HA_{\Lambda^{n+1}})^{-1}A_{\Lambda^{n+1}}^H \y.
$$
We can write $A_{\Lambda^{n+1}}$ as $A_{\Lambda^{n+1}}=\left[ A_{\Lambda^n}, A_{T^n}\right]$. Then
$$
A_{\Lambda^{n+1}}^HA_{\Lambda^{n+1}}=\left[
\begin{array}{cc}
A_{\Lambda^n}^HA_{\Lambda^n} & A_{\Lambda^n}^HA_{T^n} \\
A_{T^n}^HA_{\Lambda^n} & A^H_{T^n}A_{T^n}
\end{array}
\right].
$$
We next consider
\begin{equation}\label{eq:m1234}
(A_{\Lambda^{n+1}}^HA_{\Lambda^{n+1}})^{-1}=\left[
\begin{array}{cc}
M_1 & M_2 \\
M_3 & M_4
\end{array}
\right],
\end{equation}
where
\begin{eqnarray*}
M_4 &=& (A^H_{T^n}A_{T^n}-A_{T^n}^HA_{\Lambda^n}A_{\Lambda^n}^+A_{T^n})^{-1},\\
M_3 &=& -M_4 (A_{T^n}^HA_{\Lambda^n})(A_{\Lambda^n}^{H}A_{\Lambda^n})^{-1}.
\end{eqnarray*}
Noting (\ref{eq:m1234}) and that
$$
A_{\Lambda^{n+1}}^H\y=\left[
  \begin{array}{c}
   A_{\Lambda^n}^H\y \\
   A_{T^n}^H\y
  \end{array}
\right],
$$
we obtain that
\begin{eqnarray}
(\x^{n+1})_{T^n}&=& \left((A_{\Lambda^{n+1}}^HA_{\Lambda^{n+1}})^{-1}A_{\Lambda^{n+1}}^H \y\right)_{T^n}\nonumber\\
&=& M_3A_{\Lambda^n}^H\y+M_4A_{T^n}^H\y \nonumber\\
&=& -M_4 (A_{T^n}^HA_{\Lambda^n}) (A_{\Lambda^n}^HA_{\Lambda^n})^{-1}A_{\Lambda^n}^H\y+M_4A_{T^n}^H\y\nonumber\\
&=& M_4A_{T^n}^H(-A_{\Lambda^n}(A_{\Lambda^n}A_{\Lambda^n})^{-1}A_{\Lambda^n}^H\y+\y)\nonumber\\
&=& M_4A_{T^n}^H\left(-A_{\Lambda^n}A_{\Lambda^n}^+\y+\y \right)\nonumber\\
&=& M_4A_{T^n}^H \left(\y-A\x^n\right)\nonumber\\
&=& M_4\V^n.\label{eq:xm4vn}
\end{eqnarray}
Combining (\ref{eq:Axn2}) and (\ref{eq:xm4vn}) we have
\begin{equation}\label{eq:vmv}
(\V^n)^HM_4\V^n=\left<\V^n,(\x^{n+1})_{T^n}\right>= \|A(\x^{n+1}-\x^n)\|_2^2.
\end{equation}
To this end, we consider $\bu^HM_4^{-1}\bu$ for any $\bu\in \C^{t}$.  Note that
\begin{eqnarray}
\bu^HM_4^{-1}\bu&=& \bu^HA^H_{T^n}A_{T^n}\bu-\bu^HA_{T^n}^HA_{\Lambda^n}A_{\Lambda^n}^+A_{T^n}\bu\nonumber\\
&=&\bu^HA^H_{T^n}A_{T^n}\bu-\left<A_{T^n}\bu, P_{A_{\Lambda^n}} (A_{T^n}\bu)\right>\nonumber\\
&=& \|A_{T^n}\bu\|_2^2-\|P_{A_{\Lambda^n}}(A_{T^n}\bu)\|_2^2\label{eq:um4u}\\
&\leq& \|A_{T^n}\bu\|_2^2\nonumber\\
&\leq & (1+\delta_t)\|\bu\|_2^2,\label{eq:M4inv}
\end{eqnarray}
where $P_{A_{\Lambda^n}}(A_{T^n}\bu)$ denotes the orthogonal projection of $A_{T^n}\bu$ in the subspace ${\rm span}(A_{\Lambda^n})$.
The last inequality follows from the RIP property of $A$.
Since ${\rm Spark}(A)>\#\Lambda^{n+1}$, we have $A_{T^n}\bu\notin {\rm span}(A_{\Lambda^n})$ which implies that
 $\|P_{A_{\Lambda^n}}(A_{T^n}\bu)\|_2^2<\|A_{T^n}\bu\|_2^2$ provided $\bu\neq 0$. And hence, accoridng to (\ref{eq:um4u}),
$$
 \bu^HM_4^{-1}\bu=\|A_{T^n}\bu\|_2^2-\|P_{A_{\Lambda^n}}(A_{T^n}\bu)\|_2^2,
$$
which implies that $M_4$ is a positive-definite matrix since $\bu^HM_4^{-1}\bu>0$ provided $\bu \neq 0$. Combining \eqref{eq:vmv} and \eqref{eq:M4inv}, we obtain that
$$
 \|A(\x^{n+1}-\x^n)\|_2^2\,=\,(\V^n)^HM_4\V^n\, \geq \, \frac{1}{1+\delta_t}\|\V^n\|_2^2.
$$
Then the \eqref{eq:nn} implies that
\begin{eqnarray*}
\|\y-A\x^{n+1}\|_2^2 &=&\|\y-A\x^{n}\|_2^2-\|A(\x^{n+1}-\x^n)\|_2^2\\
&\leq & \|\y-A\x^{n}\|_2^2-\frac{1}{1+\delta_t}\|\V^n\|_2^2.
\end{eqnarray*}
\end{proof}

\begin{lem}\label{le:1}
Consider $\OMMP(M)$ and $\Lambda^n\subset \Lambda^{n+1}\subset\{1,\ldots,N\}$.  Set $T^n:=\Lambda^{n+1}\setminus \Lambda^{n}$ and $t:=\# T^n$. Suppose that the sampling
matrix $A\in \C^{m\times N}$ whose columns $a_1,\ldots,a_N$ are $\ell_2$-normalized. Then
for any $\bu\in \C^N$ whose support $U:={\rm supp}(\bu)$ not included in $\Lambda^n$, we have
\begin{eqnarray*}
\|\V^n\|_2^2\,\,  \geq \,\, \frac{\|A(\bu-\x^n)\|_2^2\left(\|\y-A\x^n\|_2^2-\|\y-A\bu\|_2^2\right)}{\|\bu_{\overline{\Lambda^n}}\|_{t,1}^2},
\end{eqnarray*}
where $\V^n:=A_{T^n}^H(\y-A\x^n)$.
\end{lem}

\begin{proof}
To this end, we only need prove that
\begin{eqnarray*}
\|\V^n\|_2^2\cdot  \|\bu_{\overline{\Lambda^n}}\|_{t,1}^2
 \geq  \|A(\bu-\x^n)\|_2^2\cdot \left(\|\y-A\x^n\|_2^2-\|\y-A\bu\|_2^2\right).
\end{eqnarray*}
When
$$
\|\y-A\x^n\|_2^2-\|\y-A\bu\|_2^2<0,
$$
the conclusion holds. So, we only consider the case where
$$
\|\y-A\x^n\|_2^2-\|\y-A\bu\|_2^2\geq 0.
$$
Recall that $T^n$ is the $t$ indices corresponding to the largest  magnitude entries in the vector  $(A^H(\y-A\x^n))_{\overline{\Lambda^n}}$. Then
$$
\|\V^n\|_2\,\,  \geq\,\,  \|(A^H(\y-A\x^n))_{\overline{\Lambda^n}}\|_{t,\infty}.
$$
Noting that $(\x^n)_{\overline{\Lambda^n}}=0$ and $(A^H(\y-A\x^n))_{{\Lambda^n}}=0$, we have
\begin{eqnarray*}
\|\V^n\|_2\cdot  \|\bu_{\overline{\Lambda^n}}\|_{t,1}
  &\geq& \|(A^H(\y-A\x^n))_{\overline{\Lambda^n}}\|_{t,\infty}\cdot  \|(\bu-\x^n)_{\overline{\Lambda^n}}\|_{t,1} \\
&\geq & {\mathcal R} \left(\left<(\bu-\x^n)_{\overline{\Lambda^n}}, (A^H(\y-A\x^n))_{\overline{\Lambda^n}}\right>\right)\\
&=&{\mathcal R} \left(\left<(\bu-\x^n), A^H(\y-A\x^n)\right>\right)\\
&=&{\mathcal R} \left(\left<A(\bu-\x^n), \y-A\x^n\right>\right) \\
&=& \frac{1}{2}\left(\|A(\bu-\x^n)\|_2^2+\|\y-A\x^n\|_2^2\right.
\left.-\:\|A(\bu-\x^n)-(\y-A\x^n)\|_2^2\right) \\
&=&\frac{1}{2} \left(\|A(\bu-\x^n)\|_2^2+\|\y-A\x^n\|_2^2-\|\y-A\bu\|_2^2\right) \\
&\geq &  \|A(\bu-\x^n)\|_2 \cdot \sqrt{ {\|\y-A\x^n\|_2^2-\|\y-A\bu\|_2^2}},
\end{eqnarray*}
which implies the result, where the second inequality follows from Lemma \ref{le:cauch}.
\end{proof}

\begin{lem}\label{th:ite}
Under the conditions of Lemma \ref{le:1}, we have
\begin{equation}
\|\y-A\x^{n+1}\|_2^2
\leq
 \|\y-A\x^n\|_2^2-
  \frac{(1-\delta)}{(1+\delta_t)\left\lceil  \frac{\# (U\setminus \Lambda^n)}{t} \right\rceil}\max\{0,
\|\y-A\x^n\|_2^2-\|\y-A\bu\|_2^2 \},\label{eq:ite}
\end{equation}
where $\delta=\delta_{\#(U\cup \Lambda^n)}$.
\end{lem}
\begin{proof}
According to Lemma \ref{th:1} and Lemma \ref{le:1}, we have
\begin{eqnarray}
\|\y-A\x^{n+1}\|_2^2\,\,& \leq& \,\, \|\y-A\x^n\|^2-\frac{1}{1+\delta_t}\|\V_n\|_2^2\nonumber\\
&\leq & \|\y-A\x^n\|^2 -\frac{\|A(\bu-\x^n)\|_2^2\left(\|\y-A\x^n\|_2^2-\|\y-A\bu\|_2^2\right)}{(1+\delta_t)\|\bu_{\overline{\Lambda^n}}\|_{t,1}^2}.\label{eq:leeq1}
\end{eqnarray}
From Lemma \ref{le:cauch}, we have
\begin{equation}\label{eq:leeq2}
\|\bu_{\overline{\Lambda^n}}\|_{t,1}^2\,\,\leq\,\, \left\lceil  \frac{\# (U\setminus \Lambda^n)}{t} \right\rceil\cdot \|\bu_{\overline{\Lambda^n}}\|_2^2.
\end{equation}
Also,
\begin{eqnarray}
\|A(\bu-\x^n)\|_2^2 &\geq& (1-\delta)\|\bu-\x^n\|_2^2 \nonumber\\
&\geq& (1-\delta)\|(\bu-\x^n)_{\overline{\Lambda^n}}\|_2^2\nonumber\\
 &\geq& (1-\delta) \|\bu_{\overline{\Lambda^n}}\|_2^2.\label{eq:leeq3}
\end{eqnarray}
Putting (\ref{eq:leeq1}), (\ref{eq:leeq2}) and (\ref{eq:leeq3}) together, we arrive at the conclusion.
\end{proof}

\begin{remark}
 Lemma \ref{th:ite} extends some results in \cite{foucartomp}, where  Foucart considered the case with $t=\# (\Lambda^{n+1}\setminus \Lambda^n)=1$, to the general case.  In fact, if  takes $t=1$ in Lemma \ref{th:ite}, one can obtain Lemma 4 in \cite{foucartomp}.
\end{remark}

\section{Proof of Theorem \ref{th:ommpt0}}

\begin{proof}[Proof of Theorem \ref{th:ommpt0}]
To state conveniently, we set  $\x':=\x_{\overline{\Lambda^0}}$ and $\bar{K}:=\max\{s',\frac{8}{M}s'\}$.
We claim that the conclusion follows provided $S\subset \Lambda^{\bar K}$.  Indeed,  since
$$\#\Lambda^{\bar K}\leq \max\{Ms', 8s'\}+\#\Lambda^0 < {\rm Spark}(A),$$
one can recover $\x$ by solving the least square, i.e.,
 $$
 \x=\aarg{{\mathbf z}: {\rm supp}({\mathbf z})\subset \Lambda^{\bar{K}}}\|\y-A {\mathbf z}\|_2
 $$
Thus, to this end, we only need prove that $S\subset \Lambda^{\bar K}$, i.e.
$\# (S\setminus \Lambda^{\bar K})=0$.
 The proof is by induction  on $s'=\#(S\setminus \Lambda^0)$. If ${s'}=0$,
  then the conclusion holds. For the induction step, we assume that the result
 holds up to an integer $s'-1$. We next show that it holds for $s'$.

  Without loss of generality, we suppose that
$$
\abs{\x'_1}\geq \abs{\x'_2}\geq \cdots \geq \abs{\x'_{s'}} > 0.
$$
For $\ell=1,\ldots,\max\{0,\lceil\log_2\frac{s'}{M}\rceil\}+1$, we set
$$
\tilde{\x}^\ell_j:=
 \begin{cases}
\x'_j & \hbox{ if $j\geq 2^{\ell-1}\cdot M+1$, }\\
0& \hbox{else,}
\end{cases}
$$
and $\tilde{\x}^0:=\x'$.
 Suppose that $L\in \Z$ such that
\begin{equation}\label{eq:L11t0}
\|\tilde{\x}^0\|_2^2 < \mu \|\tilde{\x}^1\|_2^2,\ldots,\|\tilde{\x}^{L-2}\|_2^2 <\mu \|\tilde{\x}^{L-1}\|_2^2
\end{equation}
and
\begin{equation}\label{eq:L21t0}
\|\tilde{\x}^{L-1}\|_2^2 \geq \mu \|\tilde{\x}^L\|_2^2.
\end{equation}
And hence, $L$ is the least integer such that $\|\tilde{\x}^{L-1}\|_2^2 \geq \mu \|\tilde{\x}^L\|_2^2$ and we will choose $\mu>2$
  late. The existence of such a $L$ can follow from
$\|\tilde{\x}^\ell\|_2=0$ when $\ell=\max\{0,\lceil\log_2\frac{s'}{M}\rceil\}+1$.
 And hence, we have
 $$1 \leq L\leq \max\left\{0,\left\lceil\log_2\frac{s'}{M}\right\rceil\right\}+1.$$

We first consider the  case where $L=1$.
We take $\bu=\bu^1:=\x-\tilde{\x}^1$ and $t=M$ in (\ref{eq:ite}). Then a simple observation is that
$$
\# ({\rm supp}(\bu^1)\setminus \Lambda^0)=\min\{M, s'\}.
$$
Noting that $\lceil\frac{\# ({\rm supp}(\bu^1)\setminus \Lambda^0)}{M}\rceil=1$ and
$$
\|\y-A\bu^1\|_2^2=\|A\x-A\bu^1\|_2^2=\|A\tilde{\x}^1\|_2^2.
$$
By subtracting  $\|\y-A\bu^1\|_2^2=\|A\tilde{\x}^1\|_2^2$ on both sides of (\ref{eq:ite}), we can obtain that
\begin{eqnarray*}
\max\{0,\|\y-A\x^{1}\|_2^2-\|A\tilde{\x}^1\|_2^2\}
 \leq \left(1-\frac{1-\delta_{s }}{1+\delta_{s}}\right)\max\{0,\|y-A\x^0\|_2^2-\|A\tilde{\x}^1\|_2^2 \},
\end{eqnarray*}
which implies that
\begin{eqnarray}
{\|\y-A\x^1\|_2^2}
& \leq& \left(1-\frac{1-\delta_{s}}{1+\delta_s}\right)\max\{0, \|\y-A\x^0\|_2^2-\|A\tilde{\x}^1\|_2^2\}+\|A\tilde{\x}^1\|_2^2 \nonumber\\
&=&\left(1-\frac{1-\delta_{s}}{1+\delta_s}\right)\max\{0, \|A
\tilde{\x}^0\|_2^2-\|A\tilde{\x}^1\|_2^2\}+\|A\tilde{\x}^1\|_2^2 \nonumber\\
& \leq&  \left(1-\frac{1-\delta_{s}}{1+\delta_s}\right)\|A
\tilde{\x}^0\|_2^2+\|A\tilde{\x}^1\|_2^2 \nonumber\\
& \leq&  (1+\delta_s)\left(\left(1-\frac{1-\delta_{s}}{1+\delta_s}\right)\|
\tilde{\x}^0\|_2^2+\|\tilde{\x}^1\|_2^2\right) \nonumber\\
& \leq&  2\delta_s \|\tilde{\x}^0\|_2^2+\frac{1+\delta_s}{\mu}\|\tilde{\x}^0\|_2^2
 =  \left(2\delta_s +\frac{1+\delta_s}{\mu}\right)\|\tilde{\x}^0\|_2^2,\label{eq:com1}
\end{eqnarray}
where the last inequality uses the fact that $L=1$ and hence
$\|\tilde{\x}^1\|_2^2 \leq \|\tilde{\x}^0\|_2^2 /\mu$.
On the other hand, we note that
\begin{eqnarray}
\|\y-A\x^1\|_2^2 &=& \|A(\x-\x^1)\|_2^2\nonumber\\
& \geq& (1-\delta_{2s})\|\x-\x^1\|_2^2\nonumber\\
 &\geq& (1-\delta_{2s})\|\x_{\overline{\Lambda^1}}\|_2^2.\label{eq:com2}
\end{eqnarray}
Then, combining  (\ref{eq:com1}) and (\ref{eq:com2}),  we obtain that
\begin{eqnarray*}
\|\x_{\overline{\Lambda^1}}\|_2^2 \leq \frac{1}{1-\delta_{2s}}\left(2\delta_s +\frac{1+\delta_s}{\mu}\right)\|\tilde{\x}^0\|_2^2.
\end{eqnarray*}
Noting $\delta_s\leq \delta_{2s}\leq \delta_{9s}\leq  \frac{1}{10}$, we have
$$
\frac{1+\delta_s}{1-3\delta_{2s}}\leq 2 <\mu,
$$
which implies that
$$
 \frac{1}{1-\delta_{2s}}\left(2\delta_s +\frac{1+\delta_s}{\mu}\right)<1.
 $$
 And hence,
 $$
\|\x_{\overline{\Lambda^1}}\|_2^2\,\, <\,\, \|\tilde{\x}^0\|_2^2,
$$
i.e.
$$
 \#(S\setminus \Lambda^1)\,\, \leq\,\, s'-1.
 $$
 Now we continue the algorithm with the initial feature set $\Lambda^1$.
According to  the induction assumption, we can  recover the $s$-sparse signal $\x$ within
$ \max\{s'-1, \frac{8}{M}(s'-1)\}$ iterations provided the initial feature set is $\Lambda^1$.
Thus, if one chooses the initial feature set  as $\Lambda^0$ then $\x$ can be recovered within
 $1+\max\{s'-1, \frac{8}{M}(s'-1)\}$ iterations. Then, the conclusion follows since
$$
1+\max\left\{s'-1, \frac{8}{M}(s'-1)\right\} \leq \max\left\{s', \frac{8}{M}s'\right\}.
$$

We next consider the case where $L\geq 2$.
 We take $\bu=\bu^\ell:=\x-\tilde{\x}^\ell$ and $t=M$ in (\ref{eq:ite}). Then a simple observation is that
$$\#{\rm supp}(\bu^\ell)=\#({\rm supp}(\bu^\ell)\cap \Lambda^0)+\min\{ 2^{\ell-1}M,s'\}.$$
Thus, for any $n\geq 0$,
\begin{eqnarray*}
\#({\rm supp}(\bu^\ell)\setminus \Lambda^n)
&=& \#({\rm supp}(\bu^\ell)\cap \Lambda^0)+\min\{2^{\ell-1}M,s'\} -\#({\rm supp}(\bu^\ell)\cap \Lambda^n)\\
&\leq& \min\{2^{\ell-1}M, s'\}.
\end{eqnarray*}
To state conveniently, we set
$$
\bar{U}^\ell:=\left\lceil\frac{\min\{2^{\ell-1} M,s'\}}{M}\right\rceil\in \Z.
$$
If ${\rm supp}(\bu^\ell)\not\subset \Lambda^n$ then
  we obtain that
\begin{eqnarray}
& &\max\{0,\|\y-A\x^{n+1}\|_2^2-\|A\tilde{\x}^\ell\|_2^2\}\nonumber \\
 & &\leq
  \left(1-\frac{1-\delta_{s+nM }}{(1+\delta_{M})\cdot\bar{U}^\ell}\right)\max\{0,\|\y-A\x^n\|_2^2-\|A\tilde{\x}^\ell\|_2^2 \} \nonumber\\
& &\leq   \exp\left(-\frac{1-\delta_{s+nM }}{(1+\delta_{M})\cdot\bar{U}^\ell}\right)\max\{0,
 \|\y-A\x^n\|_2^2-\|A\tilde{\x}^\ell\|_2^2 \} ,\label{eq:itexpt0}
\end{eqnarray}
which follows by subtracting
$$\|y-A\bu^\ell\|_2^2=\|A\x-A\bu^\ell\|_2^2=\|A\tilde{\x}^\ell\|_2^2$$ on both sides of (\ref{eq:ite}) in Lemma \ref{th:ite}. For the case
${\rm supp}(\bu^\ell)\subset \Lambda^n$, (\ref{eq:itexpt0}) still holds since both sides of (\ref{eq:itexpt0}) are equal to $0$.
Iterating (\ref{eq:itexpt0}) $k$ times leads to
\begin{eqnarray}\label{eq:ite00}
& &\max\{0,\|\y-A\x^{n+k}\|_2^2-\|A\tilde{\x}^\ell\|_2^2\}\\\nonumber
& &\leq  \exp\left(-k\frac{1-\delta_{s+nM}}{(1+\delta_{M})\cdot\bar{U}^\ell}\right)\max\{0,
 \|\y-A\x^n\|_2^2-\|A\tilde{\x}^\ell\|_2^2 \}
\end{eqnarray}
which implies that
\begin{eqnarray}
 & &\|\y-A\x^{n+k}\|_2^2\nonumber\\
& &\leq  \exp\left(-k\frac{1-\delta_{s+nM}}{(1+\delta_{M})\cdot\bar{U}^\ell}\right)\max\{0,
 \|\y-A\x^n\|_2^2-\|A\tilde{\x}^\ell\|_2^2 \}+\|A\tilde{\x}^\ell\|_2^2 \nonumber\\
& &\leq  \exp\left(-k\frac{1-\delta_{s+nM}}{(1+\delta_{M})\cdot\bar{U}^\ell}\right)\|\y-A\x^n\|_2^2+\|A\tilde{\x}^\ell\|_2^2.\label{eq:itexp1t0}
\end{eqnarray}
 Here, if the left side of (\ref{eq:ite00}) is $0$, then
$$\|\y-A\x^{n+k}\|_2^2\leq \|A\tilde{\x}^\ell\|_2^2.$$ Thus, (\ref{eq:itexp1t0}) still holds since $\|\y-A\x^n\|_2^2\geq 0$.
To state conveniently, for $\ell=1\ldots,L$, we set $k_\ell:=\bar{k}\cdot\bar{U}^\ell$, $k_0:=0$, $K:=k_1+\cdots+k_L$ and $\nu:=\exp\left(-\bar{k}\frac{1-\delta_{s+KM}}{1+\delta_{M}}\right)$, and we will choose $\bar{k}$ late.
For $\ell=1,\ldots, L$, we take $n:=k_0+\cdots+k_{\ell-1}$ and $k:=k_\ell$ in  (\ref{eq:itexp1t0}) and arrive at
\begin{eqnarray}\label{eq:bude}
\|\y-A\x^{k_1+\cdots+k_\ell}\|_2^2&\leq&  \exp\left(-{\bar k}\frac{1-\delta_{s+(k_0+\cdots+k_{\ell-1})M}}{1+\delta_{M}}\right)\|\y-A\x^{k_1+\cdots+k_{\ell-1}}\|_2^2 +\|A\tilde{\x}^\ell\|_2^2 \nonumber\\
&\leq& \nu\|\y-A\x^{k_1+\cdots+k_{\ell-1}}\|_2^2 +\|A\tilde{\x}^\ell\|_2^2.
\end{eqnarray}
Then, using the inequality (\ref{eq:bude}) for $L$ times,  we can obtain that
\begin{eqnarray*}
\|\y-A\x^{K}\|_2^2
 & \leq& \nu^L\|\y-A\x^0\|_2^2+\nu^{L-1}\|A\tilde{\x}^{1}\|_2^2+\cdots+\nu\|A\tilde{\x}^{L-1}\|_2^2+\|A\tilde{\x}^L\|_2^2\\
& \leq& \nu^L\|A\tilde{\x}^0\|_2^2+\cdots+\nu\|A\tilde{\x}^{L-1}\|_2^2+\|A\tilde{\x}^L\|_2^2.
\end{eqnarray*}
Here, for the second relation, we use the fact of
\begin{eqnarray*}
\|\y-A\x^0\|_2^2 =\min_{{\rm supp}({\mathbf z})\subset \Lambda^0}\|\y-A{\mathbf z}\|_2^2
\leq \|\y-A(\x-\tilde{\x}^0)\|_2^2 =\|A\tilde{\x}^0\|_2^2
\end{eqnarray*}
with ${\rm supp}(\x-\tilde{\x}^0)\subset \Lambda^0$.
Combining RIP property of $A$, (\ref{eq:L11t0}) and (\ref{eq:L21t0}), we obtain that
$$
\|A\tilde{\x}^\ell\|_2^2 \leq (1+\delta_s)\|\tilde{\x}^\ell\|_2^2\leq (1+\delta_s)\mu^{L-1-\ell}\|\tilde{\x}^{L-1}\|_2^2
$$
 for $\ell=0,1,\ldots,L$.
Note that
\begin{eqnarray}
\|\y-A\x^K\|_2^2
&\leq& \sum_{\ell=0}^L \nu^{L-\ell}\|A\tilde{\x}^\ell\|_2^2 \nonumber \\
&\leq& \frac{(1+\delta_s)\|\tilde{\x}^{L-1}\|_2^2}{\mu}\sum_{\ell=0}^L (\mu \nu)^{L-\ell}\nonumber\\
&\leq & \frac{(1+\delta_s)\|\tilde{\x}^{L-1}\|_2^2}{\mu(1-\mu \nu)},\label{eq:leq1t0}
\end{eqnarray}
and
\begin{eqnarray} \label{eq:leq2t0}
\|\y-A\x^K\|_2^2 &\geq &\|A(\x-\x^K)\|_2^2\nonumber\\
 &\geq& (1-\delta_{s+K\cdot M})\|\x-\x^K\|_2^2\\ &\geq& (1-\delta_{s+K\cdot M})\|\x_{\overline{\Lambda^K}}\|_2^2.\nonumber
\end{eqnarray}
Combining (\ref{eq:leq1t0}) and (\ref{eq:leq2t0}), we have
\begin{equation}\label{eq:di1}
\|\x_{\overline{\Lambda^K}}\|_2^2\leq  \frac{(1+\delta_s)}{(1-\delta_{s+K\cdot M})\mu(1-\mu\nu)} \|\tilde{\x}^{L-1}\|_2^2.
\end{equation}
We can choose $\bar{k}=2$,  $\mu=\frac{1}{2\nu}$, and $\delta_{s+K\cdot M}\leq \delta_{9s}\leq \frac{1}{10}$ with
$$
K=k_1+\cdots+k_L\leq 2^L \bar{k} \leq 8\frac{s}{M}.
$$
Noting that  $\nu\leq \exp(-18/11)$ and $\mu=\frac{1}{2\nu}>2$, we have
\begin{equation}\label{eq:di2}
 \frac{(1+\delta_s)}{(1-\delta_{s+K\cdot M})\mu(1-\mu \nu)}\,\, <\,\,1.
\end{equation}
Combining (\ref{eq:di1}) and (\ref{eq:di2}), we obtain that
$$
\|\x_{\overline{\Lambda^K}}\|_2^2< \|\tilde{\x}^{L-1}\|_2^2.
$$
As a result, after $K$ iterations, we have
$$
\# (S\setminus \Lambda^K)\leq  {\rm supp}(\tilde{\x}^{L-1})-1 =
s'-2^{L-2}\cdot M-1,
$$
with
$$
K=k_1+\cdots+k_L\leq 2^L \bar{k} .
$$
Now we continue the algorithm with the initial feature set $\Lambda^K$.
According to  the induction assumption, we can  recover the $s$-sparse signal $\x$ within
$ {\bar n}$ iterations provided the initial feature set is $\Lambda^K$, where
$$
\bar{n}=
\max\{s'-2^{L-2}\cdot M-1,\frac{8}{M}(s'-2^{L-2}\cdot M-1)\}.
$$
Thus, if one chooses the initial feature set  as $\Lambda^0$ then $\x$ can be recovered within
 $K+\bar{n}$ iterations. Then, the conclusion follows since
 $K+\bar{n}\leq \max\{s', \frac{8}{M}s'\}$.

\end{proof}

\section{Proofs of Theorem \ref{th:ommpt1} and  Theorem \ref{th:ommpt2} }

To prove Theorem \ref{th:ommpt1} and  Theorem \ref{th:ommpt2}, we first introduce two lemmas.

\begin{lem}\label{le:slowdecay}
Consider the $\OMMP(M)$ algorithm with  $1\leq M\leq s$. Suppose that the sampling
matrix $A\in \C^{m\times N}$  satisfies $(9s, \frac{1}{10})$-RIP.
Suppose that $\x\in \Sigma_s$, $S={\rm supp}(\x)$.
 Then
 $$\# (S\setminus\Lambda^{\bar{K}})=0,$$
  where ${\bar K}:=\myfloor{8\frac{s'}{M}+8(C_0^2+1)M}$, $s':=\#(S\setminus \Lambda^0)$
 and $
C_0={\max\limits_{j\in S }\abs{\x_j}}/{\min\limits_{j\in S}\abs{\x_j}}.
$
\end{lem}
\begin{proof}
To state conveniently, we set  $$\x':=\x_{\overline{\Lambda^0}}$$ and
$$C_2:=\frac{C_0^2}{\mu-1}+1.$$ We will choose $\mu>2$ late so that $C_2<C_0^2+1$.
To this end, we will prove that $\# (S\setminus\Lambda^{K_1})=0 $ with
$K_1=\myfloor{8\frac{s'}{M}+8C_2M}$, which implies the result.   The proof is by induction  on $s'=\#(S\setminus \Lambda^0)$. We first consider the case where $s'\leq C_2M$. According to Theorem \ref{th:ommpt0}, $\OMMP(M)$
  recover the $s$-sparse signal within $8C_2M<8(C_0^2+1)M$ iterations. Thus, we arrive at the result provided $s'\leq C_2M$.

We next consider the case where  $s'> C_2 M$.  Without loss of generality, we suppose that
$$
\abs{\x'_1}\geq \abs{\x'_2}\geq \cdots \geq \abs{\x'_{s'}} > 0.
$$
To state coneniently, for $\ell=1,\ldots,\lceil\log_2(\frac{s'}{M} )\rceil+1$, we set
$$
\tilde{\x}^\ell_j:=
 \begin{cases}
\x'_j & \hbox{ if $2^{\ell-1}{M}+1\leq j$, }\\
0& \hbox{else.}
\end{cases}
$$
and $\tilde{\x}^0:=\x'$.
 Suppose that $L\in \Z$ such that
\begin{equation}\label{eq:L1}
\|\tilde{\x}^0\|_2^2 < \mu \|\tilde{\x}^1\|_2^2,\ldots,\|\tilde{\x}^{L-2}\|_2^2 <\mu \|\tilde{\x}^{L-1}\|_2^2
\end{equation}
and
\begin{equation}\label{eq:L2}
\|\tilde{\x}^{L-1}\|_2^2 \geq \mu \|\tilde{\x}^L\|_2^2.
\end{equation}
And hence, $L$ is the least integer such that $\|\tilde{\x}^{L-1}\|_2^2 \geq \mu \|\tilde{\x}^L\|_2^2$.  The existence of such a $L$ can follow from
$\|\tilde{\x}^\ell\|_2=0$ when $\ell=\lceil\log_2\frac{s' }{M}\rceil+1$.
We next show that the assumption of $s'> C_2 M$ implies that
$\|\tilde{\x}^0\|_2^2 < \mu \|\tilde{\x}^1\|_2^2$ and hence $L\geq 2$.
Indeed, $\|\tilde{\x}^0\|_2^2 < \mu \|\tilde{\x}^1\|_2^2$ is equivelent to
\begin{equation}\label{eq:bx}
\x_1'^2+\cdots+\x_M'^2 < (\mu-1) \|\tilde{\x}^1\|_2^2.
\end{equation}
Hence, we only need argue (\ref{eq:bx}). Note that
\begin{eqnarray*}
\x_1'^2+\cdots+\x_M'^2 \leq M \max_{j\in S}\x_{j}^2 <(\mu-1)(s'-M)\min_{j\in S}\x_{j}^2\leq (\mu-1)\|\tilde{\x}^1\|_2^2,
\end{eqnarray*}
where the second relation uses the fact of
$$
s'>C_2M=\left(\frac{C_0^2}{\mu-1}+1\right)M.
$$
And hence, we have $2\leq L\leq \lceil\log_2\frac{s'}{M}\rceil+1$. We take $$\bu=\bu^\ell:=\x-\tilde{\x}^\ell$$
 and $t=M$ in (\ref{eq:ite}). Then a simple observation is that
$$
\#{\rm supp}(\bu^\ell)=\#\Lambda^0+\min\{ 2^{\ell-1}{M},s'\}.
$$
For any $n\geq 0$,
\begin{eqnarray*}
\# ({\rm supp}(\bu^\ell)\setminus \Lambda^n)
& =& \# ({\rm supp}(\bu^\ell)\cap \Lambda^0)+\min\{2^{\ell-1} {M},s'\} -\#({\rm supp}(\bu^\ell)\cap \Lambda^n)\\
&  \leq& \min\{ 2^{\ell-1} {M}, s'\}.
\end{eqnarray*}
To state conveniently, we set
$$
\bar{U}^\ell:=\left\lceil\frac{\min\{ 2^{\ell-1} {M},s'\}}{M}\right\rceil.
$$
 Noting that
 $$
 \|\y-A\bu^\ell\|_2^2=\|A\x-A\bu^\ell\|_2^2=\|A\tilde{\x}^\ell\|_2^2,
 $$
  by (\ref{eq:ite}),  we obtain that
\begin{eqnarray}
& &\max\{0,\|\y-A\x^{n+1}\|_2^2-\|A\tilde{\x}^\ell\|_2^2\}\nonumber \\
 & &\leq  \left(1-\frac{1-\delta_{s+nM}}{(1+\delta_{M})\cdot\bar{U}^\ell}\right)\max\{0,
 \|\y-A\x^n\|_2^2-\|A\tilde{\x}^\ell\|_2^2 \}\nonumber\\
& &\leq  \exp\left(-\frac{1-\delta_{s+nM}}{(1+\delta_{M})\cdot\bar{U}^\ell}\right)\max\{0,
\|\y-A\x^n\|_2^2-\|A\tilde{\x}^\ell\|_2^2 \}.\label{eq:itexp0}
\end{eqnarray}
Iterating (\ref{eq:itexp0}) for ${ k}$ times leads to
\begin{eqnarray*}
\max\{0,\|\y-A\x^{n+{ k}}\|_2^2-\|A\tilde{\x}^\ell\|_2^2\}
\leq  \exp\left(-k\frac{(1-\delta_{s+KM})}{(1+\delta_{M})\cdot\bar{U}^\ell}\right)\max\{0,\|\y-A\x^n\|_2^2-\|A\tilde{\x}^\ell\|_2^2 \},
\end{eqnarray*}
which implies that
\begin{eqnarray}
\|\y-A\x^{n+{ k}}\|_2^2
\leq  \exp\left(-k\frac{(1-\delta_{s+KM})}{(1+\delta_{M})\cdot\bar{U}^\ell}\right)\|\y-A\x^n\|_2^2+\|A\tilde{\x}^\ell\|_2^2\label{eq:itexp10}
\end{eqnarray}
where ${ k}$ and $K$ are  integers satisfying $K\geq n+{ k}$.

To state conveniently, for $\ell=1\ldots,L$, we set $k_\ell:=\bar{k}\cdot\bar{U}^\ell$, $K:=k_1+\cdots+k_L$ and $$v:=\exp\left(-\bar{k}\frac{1-\delta_{s+KM}}{1+\delta_{M}}\right),$$
 and we will choose $\bar{k}$ late.  We use (\ref{eq:itexp10})
 and a similar argument in the proof of Theorem \ref{th:ommpt0} to obtain that
\begin{eqnarray}
\|\y-A\x^K\|_2^2&\leq& \sum_{\ell=0}^L v^{L-\ell}\|A\tilde{\x}^\ell\|_2^2\nonumber\\
&\leq& \frac{(1+\delta_s)\|\tilde{\x}^{L-1}\|_2^2}{\mu}\sum_{\ell=0}^L (\mu v)^{L-\ell}\nonumber\\
& \leq & \frac{(1+\delta_s)\|\tilde{\x}^{L-1}\|_2^2}{\mu(1-\mu v)}.\label{eq:leq101}
\end{eqnarray}
Note that
\begin{eqnarray} \label{eq:leq201}
\|\y-A\x^K\|_2^2 &\geq& \|A(\x-\x^K)\|_2^2 \\
&\geq &(1-\delta_{s+KM})\|\x-\x^K\|_2^2 \nonumber\\
& \geq &(1-\delta_{s+KM})\|\x_{\overline{\Lambda^K}}\|_2^2.\nonumber
\end{eqnarray}
Combining (\ref{eq:leq101}) and (\ref{eq:leq201}), we arrive at
$$
\|\x_{\overline{\Lambda^K}}\|_2^2\leq  \frac{(1+\delta_s)}{(1-\delta_{s+KM})\mu(1-\mu v)} \|\tilde{\x}^{L-1}\|_2^2.
$$
We can choose $\bar{k}=2$,  $\mu=\frac{1}{2v}$, and
$$\delta_{s+KM}\leq \delta_{9s}\leq \frac{1}{10},$$
 and therefore  $v\leq \exp(-18/11)$ and
$$\mu =\frac{1}{2\nu}>2.$$
Here, we use $s+KM\leq s+4\bar{k}s'\leq 9s$ since
\begin{eqnarray*}
K&=&k_1+\cdots+k_L\,\leq\, \bar{k}(1+\cdots+2^{L-1})\\
&\leq& 2^L\bar{k}\,\leq\, 4\cdot \bar{k} \cdot \frac{s'}{M}.
\end{eqnarray*}
 Then
$$
\frac{(1+\delta_s)}{(1-\delta_{s+KM})\mu(1-\mu v)}<1,
$$
which implies that
$$
\|\x_{\overline{\Lambda^K}}\|_2^2< \|\tilde{\x}^{L-1}\|_2^2.
$$
As a result, after $K$ iterations, we have
$$
\# (S\setminus \Lambda^K)< \#((S\setminus \Lambda^0)\setminus {\rm supp}(\bu^{L-1})) =s'-2^{L-2}{M}.
$$
 Now we continue the algorithm with the inital feature set  $\Lambda^K$. According to the induction
assumption, we can  recover the $s$-sparse signal $\x$ in $K+\bar{n}$ iterations where
$$\bar{n}\leq \myfloor{8\frac{s'-2^{L-2} {M}}{M}+8C_2 {M}}.$$
Note that $L\geq 2$ and
\begin{eqnarray*}
K+\bar{n}&\leq& 2^L \bar{k}+\myfloor{8\frac{s'-2^{L-2}M}{M}+8C_2M}\\
&=& 8\cdot 2^{L-2}+\myfloor{8\frac{s'}{M}+8C_2M}-8\cdot2^{L-2}\\
&=&\myfloor{8\frac{s'}{M}+8C_2M}.
\end{eqnarray*}
Then we arrive at
\begin{eqnarray*}
K+\bar{n}\leq \myfloor{8\frac{s'}{M}+8C_2M},
\end{eqnarray*}
which implies the result.
\end{proof}

\begin{lem}\label{le:smalls}
Suppose that $\x$ is $s$-sparse, $S={\rm supp}(\x)$ and
$
C_0={\max\limits_{j\in S}\abs{\x_j}}/{\min\limits_{j\in S}\abs{\x_j}}.
$
Consider the $\OMMP(M)$ algorithm with  $1\leq M\leq \frac{2}{C_0^2+2}\cdot s$. Suppose that the sampling
matrix $A\in \C^{m\times N}$ whose columns $a_1,\ldots,a_N$ are $\ell_2$-normalized, and
that $A$ satisfies $(14s,\frac{1}{10})$-RIP.  Set $s':=\#(S\setminus \Lambda^0)$ and
$$
{\bar K}:=\myceil{8\frac{s'}{M}+4\cdot \ln 2\cdot \frac{s}{M}\log_2(s'+1)}.
$$
Then
$$
\# (S\setminus\Lambda^{\bar{K}})=0.
$$
\end{lem}

\begin{proof}
To state conveniently, we set  $\x':=\x_{\overline{\Lambda^0}}$.
 The proof is by induction  on $s'=\#(S\setminus \Lambda^0)$. When $s'=0$, the conclusion holds trivially.

 Without loss of generality, we suppose that
$$
\abs{\x'_1}\geq \abs{\x'_2}\geq \cdots \geq \abs{\x'_{s'}} > 0.
$$
For convenience, for $\ell=1,\ldots,\lceil\log_2(\frac{s}{M} )\rceil+1$, we set
$$
\tilde{\x}^\ell_j:=
 \begin{cases}
\x'_j & \hbox{ if $2^{\ell-1}\frac{M}{s}s'+1\leq j$, }\\
0& \hbox{else}
\end{cases}
$$
and $\tilde{\x}^0:=\x'$.
 Similar with the proof of Lemma \ref{le:slowdecay}, suppose that  $L$ is the least integer
  such that $\|\tilde{\x}^{L-1}\|_2^2 \geq \mu \|\tilde{\x}^L\|_2^2$.  We will choose $\mu>2$ late. The assumption of
  $$M<\frac{2}{C_0^2+2}s$$
  implies that
$$
\|\tilde{\x}^0\|_2^2 < \mu \|\tilde{\x}^1\|_2^2.
$$
 And hence, we have $2\leq L\leq \lceil\log_2\frac{s}{M}\rceil+1$. We take
 $$
 \bu=\bu^\ell:=\x-\tilde{\x}^\ell
 $$
 and $t=M$ in (\ref{eq:ite}). Then a simple observation is that
$$
\#{\rm supp}(\bu^\ell)=\#{\rm supp}((\bu^\ell)\cap \Lambda^0)+\min\left\{\left\lfloor 2^{\ell-1}\frac{M}{s}s'\right\rfloor,s'\right\}.
$$
For any $n\geq 0$,
\begin{eqnarray*}
\#({\rm supp}(\bu^\ell)\setminus \Lambda^n)&=&\#{\rm supp}((\bu^\ell)\cap \Lambda^0) +
\min\left\{\left\lfloor 2^{\ell-1}\frac{M}{s}s'\right\rfloor,s'\right\}
-\#{\rm supp}((\bu^\ell)\cap \Lambda^n)\\
& \leq& \min\left\{ \left\lfloor 2^{\ell-1}\frac{M}{s}s'\right\rfloor,s'\right\}.
\end{eqnarray*}
To state conveniently, we set
$$
\bar{U}^\ell:=\left\lceil\frac{\min\{\lfloor 2^{\ell-1}\frac{M}{s}s'\rfloor,s'\}}{M}\right\rceil,$$ $k_\ell:=\bar{k}\cdot\bar{U}^\ell,
$ $K:=k_1+\cdots+k_L$ and $v:=\exp\left(-\bar{k}\frac{1-\delta_{s+KM}}{1+\delta_{M}}\right)$, and we will choose $\bar{k}$ late.  We use (\ref{eq:itexp10})
 and a similar argument in the proof of Theorem \ref{th:ommpt0} to obtain that
\begin{eqnarray}
\|\y-A\x^K\|_2^2\,\,&\leq&\,\, \sum_{\ell=0}^L v^{L-\ell}\|A\tilde{\x}^\ell\|_2^2\nonumber\\
& \leq& \frac{(1+\delta_s)\|\tilde{\x}^{L-1}\|_2^2}{\mu}\sum_{\ell=0}^L (\mu v)^{L-\ell}\nonumber\\
& \leq&  \frac{(1+\delta_s)\|\tilde{\x}^{L-1}\|_2^2}{\mu(1-\mu v)}.\label{eq:leq10}
\end{eqnarray}
Note that
\begin{eqnarray} \label{eq:leq20}
\|\y-A\x^K\|_2^2& \geq& \|A(\x-\x^K)\|_2^2 \nonumber\\
 & \geq& (1-\delta_{s+KM})\|\x-\x^K\|_2^2\nonumber\\
 & \geq &(1-\delta_{s+KM})\|\x_{\overline{\Lambda^K}}\|_2^2.
\end{eqnarray}
Combining (\ref{eq:leq10}) and (\ref{eq:leq20}), we arrive at
$$
\|\x_{\overline{\Lambda^K}}\|_2^2\leq  \frac{(1+\delta_s)}{(1-\delta_{s+KM})\mu(1-\mu v)} \|\tilde{\x}^{L-1}\|_2^2.
$$
We can choose $\bar{k}=2$,  $\mu={1}/{(2v)}$, and $\delta_{s+KM}\leq \delta_{14s}\leq {1}/{10}$. And hence $v\leq \exp(-{18}/{11})$ and $\mu ={1}/{(2\nu)}>2$.
Here, we use $s+KM\leq  13s$ with
\begin{eqnarray*}
K&=&k_1+\cdots+k_L\\
&\leq& \bar{k}(1+\cdots+2^{L-1})\frac{s'}{s}+\bar{k}L\leq 2^L\bar{k}\frac{s'}{s}+\bar{k}L\\
&\leq& 4\bar{k}\frac{s'}{M}+\bar{k}L
\leq 8\frac{s'}{M}+4+2\log_2\frac{s}{M}.
\end{eqnarray*}
 Then
$$
\frac{(1+\delta_s)}{(1-\delta_{s+KM})\mu(1-\mu v)}\,\,<\,\,1,
$$
which implies that
$$
\|\x_{\overline{\Lambda^K}}\|_2^2\,\,<\,\, \|\tilde{\x}^{L-1}\|_2^2.
$$
As a result, after $K$ iterations, we have
\begin{eqnarray*}
\# (S\setminus \Lambda^K)&\leq &\#((S\setminus \Lambda^0)\setminus {\rm supp}(\bu^{L-1}))-1 \\
&=&s'-2^{L-2}\frac{M}{s}s'-1,
\end{eqnarray*}
with
\begin{eqnarray*}
K=k_1+\cdots+k_L&\leq& \bar{k}(1+\cdots+2^{L-1})\frac{s'}{s}+\bar{k}L\\
&\leq& 2^L\bar{k}\frac{s'}{s}+\bar{k}L.
 \end{eqnarray*}
 Now we continue the algorithm from the iteration $K$. According to the induction assumption, we  have
$$
\# (S\setminus \Lambda^{K+\bar{n}})=0
$$
with
\begin{eqnarray*}
\bar{n}\leq \myceil{ 8\frac{s'-2^{L-2}\frac{M}{s}s'}{M}
 +4\cdot \ln 2\cdot\frac{s}{M}\log_2\left(s'-2^{L-2}\frac{M}{s}s'\right)}.
\end{eqnarray*}
Note that $L\geq 2$ and that
$$
2^L\bar{k}\frac{s'}{s}+8\frac{s'-2^{L-2}\frac{M}{s}s'}{M}\leq 8\frac{s'}{M}.
$$
 A simple calculation shows that
\begin{eqnarray*}
& &\bar{k}L+4\cdot \ln 2\cdot\frac{s}{M}\cdot\log_2\left(s'-2^{L-2}\frac{M}{s}s'\right)\\
& &=4\cdot \ln 2\cdot\frac{s}{M}\cdot\log_2s'+\bar{k}L+4\cdot \ln 2\cdot\frac{s}{M}\log_2\left(1-2^{L-2}\frac{M}{s}\right)\\
& &\leq 4\cdot \ln 2\cdot\frac{s}{M}\log_2(s'+1).
\end{eqnarray*}
 Then we arrive at
\begin{eqnarray*}
K+\bar{n} &\leq& 2^L\bar{k}\frac{s'}{s}+\bar{k}L+\bar{n} \\
&\leq& \myceil{8\frac{s'}{M}+4\cdot \ln 2 \cdot\frac{s}{M}\cdot\log_2(s'+1)},
\end{eqnarray*}
which implies the result.
\end{proof}

\begin{proof}[Proof of Theorem  \ref{th:ommpt1}]
According to Lemma \ref{le:slowdecay},   after $\OMMP(M)$ running ${\bar K}$ steps, we have
$$
S\,\,\subset\,\, \Lambda^{\bar K}
$$
where
$${\bar K}=\myfloor{8\frac{s'}{M}+8(C_0^2+1)M}\leq \myfloor{8(C_0^2+2)\frac{s'}{M}}.$$
Here we use the assumption of $M\leq \sqrt{s'}$.
 Since $\OMMP(M)$ chooses $M$ atoms at each iteration, we have
$$
\# \Lambda^{\bar K}\,\,\leq\,\, {\bar K}M\,\,\leq\,\, 8(C_0^2+2)s'+\#\Lambda^0.
$$
Noting that ${\rm Spark}(A)> 8(C_0^2+2){s'}+\#\Lambda^0$, we obtain that
$$
\argmin{{\mathbf z}\in \C^N, {\rm supp ({\mathbf z})}\subset\Lambda^{\bar K} }\|A{\mathbf z}-\y\|_2=\x
$$
 which implies that  $\OMMP(M)$ can recover the $s$-sparse signal $\x$ within
 $\myfloor{8(C_0^2+2)\frac{s'}{M}}$ iterations.
\end{proof}

\begin{proof}[Proof of Theorem \ref{th:ommpt2}]
By Lemma \ref{le:smalls}, we have
$$
S\,\,\subset\,\, \Lambda^{\bar K},
$$
since
 \begin{eqnarray*}
 {\bar K} &\leq&\myceil{8\frac{s}{M}+4\cdot \ln 2\cdot \frac{s}{M}\log_2(s+1)}\\
 &\leq& \myceil{8\frac{s}{M}\log_2(2(s+1))}=\myceil{\frac{8}{\alpha}\log_2(2(s+1))}.
 \end{eqnarray*}
 Here, we use the fact of $\Lambda^0=\emptyset$ and hence $\#(S\setminus \Lambda^0)=s$.
Also, noting that
$$
\#\Lambda^{\bar K}\leq {\bar K}M \leq 8 s \log_2(2(s+1))
$$
and
$$
 {\rm Spark}(A)>  8 s \log_2(2(s+1)),
$$
we obtain that
$$
\argmin{{\mathbf z}\in \C^N,\, {\rm supp ({\mathbf z})}\subset\Lambda^{\bar K} }\|A{\mathbf z}-\y\|_2=\x,
$$
 which implies  the result.
\end{proof}

\section{Proof of Theorem \ref{th:ompss}}

\begin{proof}
The proof proceed by induction. We assume that $\Lambda^\ell \subset {\rm supp}(\x)$ holds for $\ell=0,\ldots,n-1\leq s-1$. We next
consider $\Lambda^n$. Set
$$
\tilde{\x}^{n-1}:=\x_{\overline{\Lambda^{n-1}\cup \{j^{n-1}\}}},\qquad \bu:= \x_{{\Lambda^{n-1}\cup \{j^{n-1}\}}}
$$
where  $j^{n-1}$  is the indices of the largest entries of $\x_{\overline{\Lambda^{n-1}}}$ in magnitude.
Lemma \ref{th:ite} implies that
\begin{eqnarray*}
\|\y-A\x^{n}\|_2^2 \leq \|\y-A\x^{n-1}\|_2^2-
 \frac{(1-\delta_{s})} {\# (U\setminus \Lambda^{n-1})}\max\{0,\|\y-A\x^{n-1}\|_2^2-\|\y-A\bu\|_2^2 \},
\end{eqnarray*}
where $U:={\rm supp}(\bu)$. Noting that $\# (U\setminus \Lambda^{n-1})=1$, we have
\begin{eqnarray*}
\|\y-A\x^{n}\|_2^2
 & \leq& \|\y-A\x^{n-1}\|_2^2-
   (1-\delta_{n})\max\{0,\|\y-A\x^{n-1}\|_2^2-\|\y-A\bu\|_2^2 \}\\
    & \leq&  \|\y-A\x^{n-1}\|_2^2-
    (1-\delta_{n})\max\{0,\|\y-A\x^{n-1}\|_2^2-\|A\tilde{\x}^{n-1}\|_2^2 \}.
\end{eqnarray*}
We claim that
$$
\|\y-A\x^{n-1}\|_2^2\,\,\geq\,\, \|A\tilde{\x}^{n-1}\|_2^2.
$$
Then we have
\begin{eqnarray*}
\|\y-A\x^{n}\|_2^2 &\leq & \|\y-A\x^{n-1}\|_2^2-
    (1-\delta_{n})\max\{0,\|\y-A\x^{n-1}\|_2^2-\|A\tilde{\x}^{n-1}\|_2^2 \}\\
 &\leq&  \delta_n\|\y-A\x^{n-1}\|_2^2+(1-\delta_{n})\|A\tilde{\x}^{n-1}\|_2^2\\
&\leq &  \delta_s\|A(\x-\x^{n-1})\|_2^2+\|A\tilde{\x}^{n-1}\|_2^2 \\
&\leq & \delta_s (1+\delta_s)\|\x_{\overline{\Lambda^{n-1}}}\|_2^2+(1+\delta_s) \|\tilde{\x}^{n-1}\|_2^2 \\
&\leq & (1+\delta_s)\left(\delta_s+\frac{1}{\alpha^2}\right)\|\x_{\overline{\Lambda^{n-1}}}\|_2^2,
\end{eqnarray*}
Here, for the last inequality,  we use the fact of  $\|\tilde{\x}^{n-1}\|_2^2 \leq {\|\x_{\overline{\Lambda^{n-1}}}\|^2}/{\alpha^2}$ since $\x$ is $\alpha$-decaying.  On the other hand, we have
\begin{eqnarray*}
\|\y-A\x^n\|_2^2 &=& \|A(\x-\x^n)\|_2^2\,\, \geq\,\, \|A(\x-\x_{\Lambda^n})\|_2^2 \\
&\geq& \|A\x_{\overline{\Lambda^n}}\|_2^2\\
& \geq& (1-\delta_s) \|\x_{\overline{\Lambda^n}}\|_2^2.
\end{eqnarray*}
Combing the results above, we obtain that
$$
\|\x_{\overline{\Lambda^n}}\|_2^2 \leq \beta \|\x_{\overline{\Lambda^{n-1}}}\|_2^2
$$
where
$$
\beta=\frac{1+\delta_s}{1-\delta_s}\left(\delta_s+\frac{1}{\alpha^2}\right).
$$
Note that $\delta_s<\sqrt{2}-1$ and hence $2-(1+\delta_s)^2>0$. Then when
 $$
 \alpha > \sqrt{\frac{1+\delta_s}{2-(1+\delta_s)^2}},
 $$
 we have
 $$
 \beta<1.
 $$
  And hence,
$$
\|\x_{\overline{\Lambda^n}}\|_2^2 < \|\x_{\overline{\Lambda^{n-1}}}\|_2^2,
$$
which implies that ${\Lambda^n}\subset {\rm supp}(\x)$.

We remain to argue that
$$
\|\y-A\x^{n-1}\|_2^2\,\,\geq\,\, \|A\tilde{\x}^{n-1}\|_2^2.
$$
We assume that
$$
\|\y-A\x^{n-1}\|_2^2< \|A\tilde{\x}^{n-1}\|_2^2,
$$
and we shall derive a contradiction. The RIP property of the matrix $A$ implies that
\begin{eqnarray*}
(1-\delta_s)\|\x_{\overline{\Lambda^{n-1}}}\|_2^2 \leq \|\y-A\x^{n-1}\|_2^2
< \|A\tilde{\x}^{n-1}\|_2^2 \leq (1+\delta_s) \|\tilde{\x}^{n-1}\|_2^2.
\end{eqnarray*}
And hence,
$$
\|\x_{\overline{\Lambda^{n-1}}}\|_2^2 \leq \frac{1+\delta_s}{1-\delta_s}\|\tilde{\x}^{n-1}\|_2^2.
$$
Noting that $\alpha^2 \|\tilde{\x}^{n-1}\|_2^2 \leq \|\x_{\overline{\Lambda^{n-1}}}\|_2^2$, we have
$$
\alpha^2\leq \frac{1+\delta_s}{1-\delta_s},
$$
which contradicts with
$\alpha^2 > {\frac{1+\delta_s}{2-(1+\delta_s)^2}}$.
\end{proof}

 \end{appendices}

\bigskip \medskip

\noindent {\bf Authors' addresses:}

\medskip

\noindent Zhiqiang Xu,
LSEC, Institute of Computational Mathematics,
Academy of Mathematics and Systems Science
Chinese Academy of Sciences
Beijing 100190, CHINA
 {\tt Email: xuzq@lsec.cc.ac.cn}

\end{document}